\documentclass[journal, a4paper]{IEEEtran}
\makeatletter
\def\endthebibliography{%
	\def\@noitemerr{\@latex@warning{Empty `thebibliography' environment}}%
	\endlist
}
\makeatother

\usepackage{amsmath,amssymb,amsfonts}
\usepackage{graphicx}
\usepackage{comment}
\usepackage{amsmath}
\usepackage{amssymb}
\usepackage{amsthm}
\def\BibTeX{{\rm B\kern-.05em{\sc i\kern-.025em b}\kern-.08em
    T\kern-.1667em\lower.7ex\hbox{E}\kern-.125emX}}

\newcommand{\colb}{}

\usepackage{filecontents}
\usepackage{textcomp}
\usepackage{xcolor}

\newtheorem{thm}{Theorem}

\newtheorem{prop}[thm]{Proposition}

\newtheorem{rem}{Remark}
\begin{document}

\title{A queueing approach to the latency of decoupled UL/DL with flexible TDD and asymmetric services}

\author{
		\IEEEauthorblockN{}
\author{Beatriz~Soret, Petar Popovski, Kristoffer Stern\\
Connectivity Section, Department of Electronic Systems, Aalborg University\\
bsa@es.aau.dk}
	}

\maketitle
\begin{abstract}
One of the main novelties in 5G is the flexible Time Division Duplex (TDD) frame, which allows adaptation to the latency requirements. However,  this flexibility is not sufficient to support heterogeneous latency requirements, in which different traffic instances have different switching requirements between Uplink (UL) and Downlink (DL). This is visible in a traffic mix of  enhanced mobile broadband (eMBB) and ultra-reliable low-latency communications (URLLC). In this paper we address this problem through the use of a decoupled UL/DL access, where the UL and the DL of a device are not necessarily served by the same base station. The latency gain over coupled access is quantified in the form of queueing sojourn time in a Rayleigh channel, as well as an  upper bound for critical traffic.
\end{abstract} 
\section{Introduction}
The main enhancement of 5G Time Division Duplex (TDD) as compared to 4G is the flexibility in the assignment of the two directions, uplink (UL) and downlink (DL), allowing for a very agile adaptation to the instantaneous traffic variations. Another remarkable feature of 5G \colb{New Radio (NR)} is the support of three generic services with vastly heterogeneous requirements (e.g., in the packet sizes): enhanced mobile broadband (eMBB), massive machine-type communications (mMTC) and ultra-reliable low-latency communications (URLLC). The latter will enable the real time interactive applications with two-way traffic, envisioned with the emergence of tactile Internet. 

It is well known from queueing theory that waiting in a single line with two available servers is on average better than waiting in separated lines with one server each \cite{Kleinrock1975}. The intuition behind is that tasks with a long task in front of the queue shall wait for a long time if only one server is available, whereas having a second server reduces the blocking situations. Translating this principle to a \colb{TDD cellular system, where the UL and DL transmission cannot occur simultaneously, we can interpret that the UL and DL transmissions are waiting in the same queue and the transmission direction (UL/DL) of the wireless link adapts to the direction (UL/DL) of the queued packets. In this paper,} we study the gain of decoupling the UL and the DL directions under heterogeneous \colb{Time Transmission Interval} (TTI) requirements. 

\colb{The idea of decoupling the access \cite{Boccardi2016} arose in the context of Heterogeneous Networks (HetNets) with the goal of alleviating  the UL-DL asymmetry and improving the average throughput. }Being the focus on the user association and the interference, the related literature has mostly used stochastic geometry for the analysis \cite{Smiljkovikj2015}. 
Several works have looked at the interplay of time slot length versus the switching cost in 4G TDD (see e.g., \cite{ElBamby2015} \cite{Kerttula2016}). \colb{The switching time is related to the difference in distances and propagation delays among devices in the UL, and the need for a timing advance to account for such differences. The increased base station density results in smaller link distances. Thus, the switching time in 5G, especially indoors, is reduced and cannot be seen as a bottleneck anymore.} 

Another relevant research question in 4G networks has been the possibility of having a link to more than one transmission point. \colb{Coordinated Multi-Point (CoMP) transmission was introduced in 4G to allow a device to simultaneously transmit and receive data on multiple transmission points (TPs) \cite{Qamar2017}. One of the CoMP techniques, Transmission Point Selection (TPS), entails the device being dynamically scheduled by the most appropriate TP. Besides in-band CoMP, 5G NR introduces the possibility of \textit{multi-connectivity} across bands. While CoMP is typically used to improve the throughput, multi-connectivity mostly serves to improve the availability and the reliability \cite{Ohmann2016}. 
In any case, the studies have been typically limited to one of the two transmission directions, i.e., finding methods to optimize the DL or the UL. }For example, \cite{Fernandez2017} studies a DL centralized joint cell association and scheduling mechanism for eMBB traffic, based on dynamic cell switching by which users are not always served by the strongest perceived cell. 

\begin{figure}[t]
	\centering
	\includegraphics[width=0.4\textwidth]{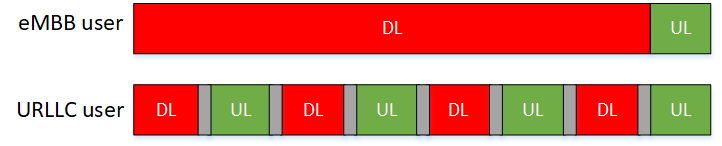}
	\caption{Traffic patterns of two-way traffic with different TTI requirements: eMBB device and URLLC interactive device. The interactive URLLC device is modeled with some processing time between transmission directions. }
	\label{fig:users}
\end{figure}
This letter proposes exploiting the extra diversity of the decoupled access to satisfy low latency requirements, and addresses UL and DL in a unified model. Each slot, of possibly different size in the frequency-time plane, can be assigned to either the UL or DL direction, depending on traffic load and received signal conditions. 
 We use a queueing model for the analysis. The reference example with heterogeneous TTI requirements is the mix of eMBB and interactive URLLC devices, \colb{as shown in Figure \ref{fig:users}}. The eMBB device requires long DL transmissions followed by short UL \colb{acknowledgement/negative acknowledgments} (ACKs/NACKs), whereas an interactive process has a stringent latency requirement and sends short UL/DL packets continuously, \colb{with a processing time between each UL and DL packet generation.}
\begin{figure}[t]
	\centering
	\includegraphics[width=0.4\textwidth]{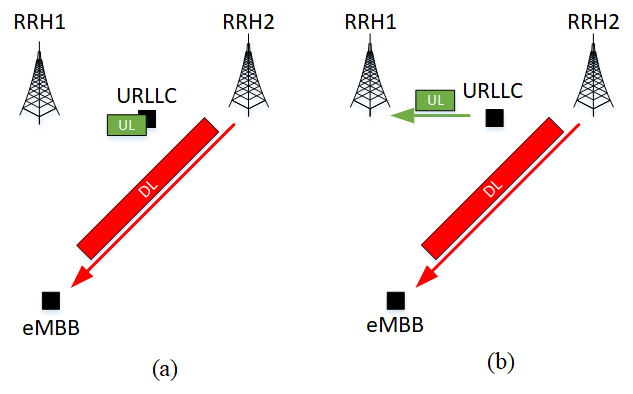}
	\caption{(a) Coupled access. The URLLC device receives UL and DL from RRH2. The UL packet has to wait until the long DL packet is transmitted. (b) Decoupled access. The URLLC device can receive UL and DL from different RRHs. The UL packet is transmitted to RRH1 while RRH2 transmits a long DL packet}
	\label{fig:decoupled}
\end{figure}

The rest of the letter is organized as follows. In Section \ref{sec:system_model} the system model is detailed. In Section \ref{sec:sojourn_time} the sojourn time is analyzed. Section \ref{sec:upper_bound} discusses an upper bound for a priority interactive URLLC user. Conclusions are in Section \ref{sec:conclusions}. 
	
	\begin{figure*}[t]
		\centering
		\includegraphics[width=1\textwidth]{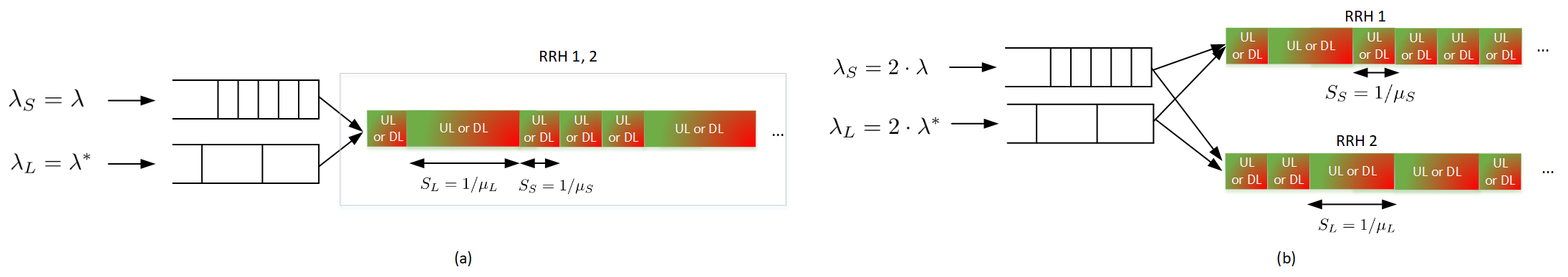}
		\caption{Queueing model with flexible TDD and devices with long and short TTI requirements. (a) Standard coupled access. Devices get the UL and DL from the same RRH. All the RRHs in the pool coordinate the transmission direction. (b) Decoupled access. Devices may receive the UL from one RRH and the DL from another one. The RRHs do not necessarily coordinate the transmission direction. 
		}
		\label{fig:queue_coupled_decoupled}
	\end{figure*}	


\section{System model} \label{sec:system_model}

We consider a TDD dense cell deployment with a central baseband pool connected with a fronthaul to a large number of small cells (Remote Radio Heads RRHs) where the Radio Frequency (RF) functionality is located. In a dense deployment each device is likely to receive a good or a fair signal quality from more than one cell. The small cells are not necessarily coordinated in their transmission directions, but opportunistically serve the traffic. The small cells serve two-way traffic from eMBB and URLLC devices like the ones in Figure \ref{fig:users}. The central unit can quickly decide which device to be served by each cell.

\colb{The benchmark is the standard coupled UL/DL in which devices are connected to only one RRH}. The base station must allocate long DL periods for the eMBB. Preemption can still be used for the small DL packets, but not for the UL, and therefore the overall latency requirement of the interactive user is challenged. The alternative is to allow decoupled UL/DL. Each device can connect to a maximum of two RRHs. Besides the primary cell $\mathcal{P}$ with the highest received power, the second one in power, denoted by $\mathcal{S}$, is reachable if it is at most $T$ dB below $\mathcal{P}$. The half-duplex devices can  be served from any of the two base stations, e.g., receive the DL from the first base station and the UL from the second one. 

The 5G NR frame has been designed with the premise of providing the necessary flexibility to support a heterogeneity of services and requirements. The main principle is that strict timing relations are avoided. On the same page, the TDD DL/UL scheme is much more flexible than LTE: a slot can contain all DL, all UL, or almost any other DL/UL ratio, and the pattern can be changed in each slot or subframe. 
The faster TDD turn-around and the self-contained concept, such that data and ACK can be scheduled in the same slot, are enablers for low-latency devices. In spite of these enhancements, we identify a limitation when the TTI requirements are asymmetric, due to the enforced transmission direction. \colb{For example, when a URLLC request arrives in the UL but the primary RRH is busy with a long eMBB transmission, the latency requirement can be met if the request is scheduled in the secondary RRH (see Figure \ref{fig:decoupled}). }

We assume a single spectrum with unit bandwidth. The instantaneous SNR is
\begin{equation}
    \gamma(t) = |h(t)|^2 \frac{E_s}{N_0}
\end{equation}

\noindent where $E_S$ is the average energy per symbol, $N_0$ the noise power spectral density, and $h(t)$ the complex channel envelope. Considering a block Rayleigh fading channel with Gaussian noise, the SNR has an  average of $\bar{\gamma}$ and it is distributed as:
\begin{equation}
    f_{\gamma} = \frac{1}{\bar{\gamma}}e^{-\gamma/\bar{\gamma}}
\end{equation}

\colb{Regarding the interference, the \textit{classical} UL-UL and DL-DL interference has been widely addressed in the context of 4G HetNets \cite{Soret2015} and later widened to 5G networks \cite{Soret2018}. Both device-based (e.g., interference cancellation receivers) and network-based (e.g., transmit power control) interference mitigation techniques are applicable to our scenario.  } Nevertheless, the lack of coordination in the transmission directions represents a major challenge  
in the form of inter-RRH and inter-device interference. A similar  problem was addressed in \cite{Popovski2015}, where the notion of interference spin was introduced for the optimization of the two-way scheduling in terms of the sum-rate. The framework can be adapted to be used in our scenario, using the latency as the \colb{Key Performance Indicator}.

With the focus on the queueing gains, we assume that the inter-RRH interference is ideally cancelled by sending the signal of the DL RRH to the UL RRH, such that it can be subtracted from the received signal. \colb{As per the inter-device interference, the challenging scenario is the reception of a DL signal from the RRH when a nearby device with line of sight (LOS) is transmitting in the UL. The situation is widely improved if there is total or partial signal obstruction between the two devices. This can be favoured by letting the scheduler prioritize the allocation of NLOS devices, as long as the latency requirements are fulfilled. } Alternatively, a parametric approach is also possible, where the average interference level is mapped to a transmission latency. This reduces the inference process to a parameter estimation problem. \colb{Another relevant consideration is that in large cell deployments, the difference in power between the DL signal and the UL signal is significant, but in small indoor cells they are of the same order. Therefore, the cross-interference UL-DL can be treated similarly as the UL-UL and DL-DL interference.}  

\section{Sojourn time} \label{sec:sojourn_time}
\colb{The analysis of the sojourn time is based on a multiclass M/G/s queue like the one in \mbox{Figure \ref{fig:queue_coupled_decoupled}}, where the traffic is separated in two queues for small and large packets}. The sub-indexes $S$ and $L$ refer to short and long TTIs, respectively. \colb{To exploit the flexible TDD, \colb{no queue is dedicated to a given transmission direction.} 
At each time instant, the transmission direction is imposed by the traffic: DL if the Head Of Line (HOL) packet is DL and UL otherwise. The server models the two-way wireless connection between RRH(s) and devices. Taking the reference traffic mix of Figure \ref{fig:users}, the short TTI queue is used by the interactive URLLC devices, whereas the eMBB devices store the long DL transmissions in the long TTI queue, and the UL ACKs/NACKs in the short TTI queue.} In the coupled UL/DL, see \mbox{Fig. \ref{fig:queue_coupled_decoupled} (a)}, each RRH is serving a short and a long TTI queue in both directions. The difference in the decoupled case is that the two queues have access to both RRHs. The queue is conservative: if the system is not empty, then the server is busy (or, in the decoupled case, at least one of the servers is busy). Moreover, there is no loss of work. For a fair comparison, the amount of traffic for the decoupled case is doubled.

\colb{Between queues, short packets have strict priority over long packets. The policy within each queue is First In First Out (FIFO). Due to the interactive nature of the URLLC traffic, modeled with a processing time between UL/DL packets, there is no need to have a special coordination between the packets scheduled in each RRH for the decoupled access. In other words, if a device has a UL packet in the queue, the consecutive DL packet will be queued only after the UL is received. This means that there is never a simultaneous transmission of a DL and UL packet from the same device. }

Being the arrival process of each queue Poisson-distributed, the total arrival process is also Poisson-distributed with rate:
\begin{equation}
    \lambda = \lambda_S + \lambda_L
\end{equation}

\noindent where the arrival rates are $\lambda_L$ and $\lambda_S$ for the long TTI and short TTI devices, respectively. 

Similarly, $\mu_L$ and $\mu_S$ denote the service rates. The inverses are the TTI duration (i.e., the service time), $S_S$ and $S_L$ ($S_S << S_L$). The time is discretized, and the minimum scheduling unit is given by the short, fixed TTI duration, $S_S$. 
Long eMBB transmissions use a discrete adaptation scheme with the range of received SNR divided into $M$ consecutive regions, each of which
is associated to a transmission rate  within the fading region $(\Gamma_{i-1}, \Gamma_i), i =
1, ..,M $. The better the channel quality, the higher the transmission rate. Thus, the service rate can take a value from a discrete set
\begin{equation}
\mu_{L}(t) =  \mu_i , \;\;\; \Gamma_{i-1} \leq \gamma(t) <  \Gamma_i , \; \; i = 1..M
\end{equation}

For Rayleigh channels, the probability of using the $i$th constellation is 
\colb{\begin{equation}
    p_i = \exp\left(-\frac{\Gamma_{i-1}}{\bar{\gamma}}\right) - \exp\left(-\frac{\Gamma_{i}}{\bar{\gamma}}\right)
\end{equation}}

The first and second moment of the service time are given by 

\begin{equation}
E[S_L] = \frac{1}{\mu_L} = \sum_{i=1}^M  \frac{p_i}{\mu_i} 
\end{equation} 
\begin{equation}
E[S^2_L] = \sum_{i=1}^M  \frac{p_i}{\mu^2_i} 
\end{equation}

The interactive URLLC devices do not have the possibility of using a closed loop and the transmission rate is fixed, i.e.
\begin{equation}
  E[S_{S}] = \frac{1}{\mu_S}, \;\;\; 
    E[S^2_S] = \frac{1}{\mu_S^2}
\end{equation}
To avoid saturation, the overall system utilization must be:
\begin{equation}
    \rho = \rho_L + \rho_S = \lambda_L E[S_L] + \lambda_S E[S_S]< 1 \label{eq:rho}
\end{equation}

Consider the $i$th data packet arriving to the system with a single server. The sojourn time comprises the queue waiting time, the frame alignment time and the transmission time. If it is a URLLC packet, it must wait in the queue for the residual time until the end of the current packet transmission plus the . If it is an eMBB packet, then it must also wait for the transmission of the URLLC packets arrived during its queueing time. 

\begin{prop}
The average sojourn time of the short and long TTI queues in the multiclass M/G/1 with priorities and discretized time is given by
\begin{equation}
\begin{aligned}
    E[T^{M/G/1}_S] = \frac{\lambda_L E[S_L^2] + \lambda_S E[S_S^2]}{2 (1-\rho_S)} &+ \frac{1}{\mu_S} + \frac{1}{2 \mu_S}, \\  
     E[T^{M/G/1}_L] = \frac{\lambda_L E[S_L^2] + \lambda_S E[S_S^2]}{2 (1-\rho) (1-\rho_S)} &+ \frac{1}{\mu_L} + \frac{1}{2 \mu_S} \label{eq:prop1}
\end{aligned}
\end{equation}

\noindent The average sojourn time of the short and long TTI queues in the multiclass M/G/2  with priorities and discretized time is approximated by
\begin{equation}
\begin{aligned}
    E[T^{M/G/2}_S] & \approx \frac{\lambda_L E[S_L^2] + \lambda_S E[S_S^2]}{(\lambda_L E[S_L] + \lambda_S E[S_S])^2} \cdot \frac{\rho^{\sqrt{6}-1}}{4\mu_L(1-\rho_S)}  \\ +  &\frac{1}{\mu_S} + \frac{1}{2 \mu_S}, \\ 
    E[T^{M/G/2}_L] & \approx \frac{\lambda_L E[S_L^2] + \lambda_S E[S_S^2]}{(\lambda_L E[S_L] + \lambda_S E[S_S])^2} \\ &\cdot \frac{\rho^{\sqrt{6}-1}}{4\mu_L(1-\rho)(1-\rho_S)} +  \frac{1}{\mu_L} + \frac{1}{2 \mu_S} \label{eq:prop1_2}
\end{aligned}
\end{equation}

\end{prop}
\begin{proof}
The result for the M/G/1 is a generalization of the Pollaczek-Khinchine formula, by considering the multi-class case, the priority and non-priority classes, and using PASTA and Little's law. The last term, $\frac{1}{2\mu_S}$, accounts for the time discretization or the frame alignment, modeled as a uniform random variable $U$ in $[0, S_S]$. 

The result for the M/G/2 is a generalization of the approximated result in \cite{Kimura1986} for GI/G/s, 
\begin{equation}
   E[W^{M/G/s}] \approx \frac{ 1+C_s^2}{2} \frac{\rho^{\sqrt{2(s+1)}-1}}{s\mu(1-\rho)} 
\end{equation}
\noindent where $C_s^2$ is the coefficient of variation of the service process \colb{and $\mu$ is the average service rate}. Then, we use the observation in \cite{Bondi1984},
\begin{equation}
   \frac{E[W^{M/GI/s/prio}]}{E[W^{M/GI/s/FCFS}]} \approx \frac{E[W^{M/GI/1/prio}]}{E[W^{M/GI/1/FCFS}]}
\end{equation}
\noindent to consider the priority and non-priority classes. 

\colb{The ratios $\frac{E[W^{M/GI/1/prio}]}{E[W^{M/GI/1/FCFS}]}$ and $\frac{E[W^{M/GI/s/prio}]}{E[W^{M/GI/s/FCFS}]}$ reflects the effect on the sojourn time of converting from FCFS scheduling to priority scheduling with single and multiple ($s$) servers, respectively. The service rate with multiple servers is $s$ times the one with one server. Based on this observation, the intuition behind is that the ratio among waiting times remains the same in the single and multiple server systems because the relationships between service time and order of selection from the queues will be the same \cite{Bondi1984} .  }
\end{proof}

The numerical evaluation of Proposition 1 and the comparison with the simulations are shown in Figure \ref{fig:sims_TTI}. \colb{The short TTI is set to 1, and the long TTI takes the values 2, 10 or 15 depending on the channel quality (with the thresholds set to 0 and 10 dB). The total sojourn time is plotted versus the system utilization $\rho$. The arrival rates are are obtained from \mbox{equation (\ref{eq:rho})} and fixing $\lambda_L=4 \cdot \lambda_S$. } The devices with long TTI spend more time in the system, due to the longer service time and the low priority. Moreover, the decoupled access reduces the average time, and the improvement is remarkable as the intensity increases, corresponding to cases in which long tasks keep the server busy with a higher probability.

\begin{figure}[t]
	\centering
	\includegraphics[width=0.5\textwidth]{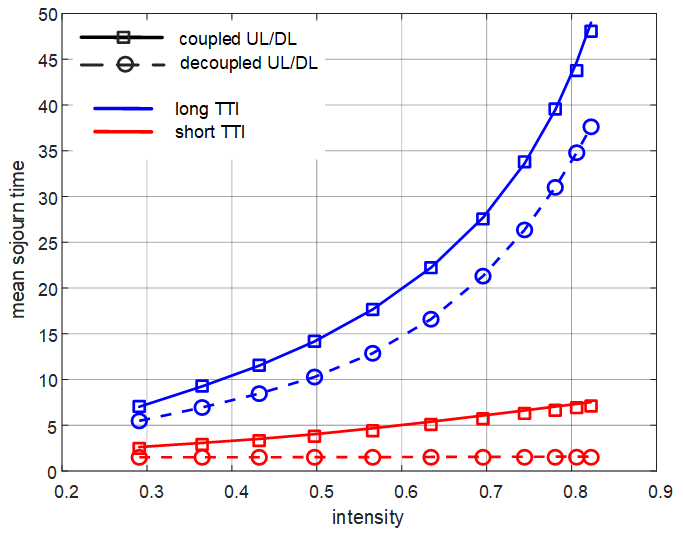}
	\caption{\colb{Comparison of the mean sojourn time with long and short TTI requirements versus the intensity $\rho$, using coupled and decoupled access. $S_S = 1$ (fixed), $S_L = 2, 10$ or $15$ (depending on the Rayleigh channel) and $\lambda_L = 4 \cdot \lambda_S$. The average sojourn time are obtained from (\ref{eq:prop1}) and (\ref{eq:prop1_2}). 
	Short packets have strict priority over long packets.}}
	\label{fig:sims_TTI}
	\end{figure}

\section{Upper bound of the cycle time: priority device} \label{sec:upper_bound}
We have studied the average gains for URLLC as a homogeneous service with FIFO policy among URLLC packets. Next, we give an upper bound of the latency distribution of the decoupled access by considering the two-way URLLC device with the highest priority. The processing time between transmission directions in the interactive traffic is $t_{proc}$, and the TTI length is $S_{S}$. In the background, there is broadband traffic with a maximum long TTI $S_{L}$ and no strict latency requirement. The interactive device has scheduling priority over any other device.

The two-way traffic is decomposed in UL-DL cycles, from the arrival instant of a UL packet to the reception of the following DL packet. Figure \ref{fig:cycle_time} shows an example of the round trip time with decoupled access. The processing time $t_{proc}$ and the transmission time $S_S$ add to the total cycle time. Moreover, both directions might find the RRHs busy, and the user has to wait the residual time $t_{res}$, i.e., the time til one of the RRHs is available. The cycle time is written

\begin{equation} \label{eq:t_cycle}
t_{cycle} = 2 \cdot (S_S  + t_{res}) + t_{proc}
\end{equation}

Assuming the same constant $S_S$ and $t_{proc}$ for the coupled and the decoupled access, the only randomness in equation (\ref{eq:t_cycle}) comes from the residual time. In our scenario, $t_{res}$ is confined to the interval $(0..S_{L}$). In the coupled scheme, the residual time in each base station follows a generic distribution $G$ between $0$ and $S_{L}$ (the longest possible TTI duration for eMBB traffic). We call this random variable $X_i \stackrel{}{\sim} G(0,S_{L})$, were $i$ is the RRH id. In the decoupled case, the residual time is the minimum between the residual times of the two RRHs,
\begin{equation}
Y  \stackrel{}{\sim}  \min(X_1, X_2)
\end{equation}

\begin{rem}
The CDF of the residual time in the decoupled access is given by the minimum between the residual times of the two RRHs, therefore

\begin{equation}
\begin{aligned}
F_Y(y) &= \text{Pr}\{Y \leq y\} =1-\left[1-F_{X_i}(y)\right]^2
\end{aligned}
\end{equation}

\end{rem}
\begin{rem}
Regardless of the distribution, the CDF of $Y$ is lower than the CDF of $X_i$, and therefore the latency of the decoupled access is always better than the coupled case. 
\end{rem}

\begin{figure}[t]
	\centering
	\includegraphics[width=0.5\textwidth]{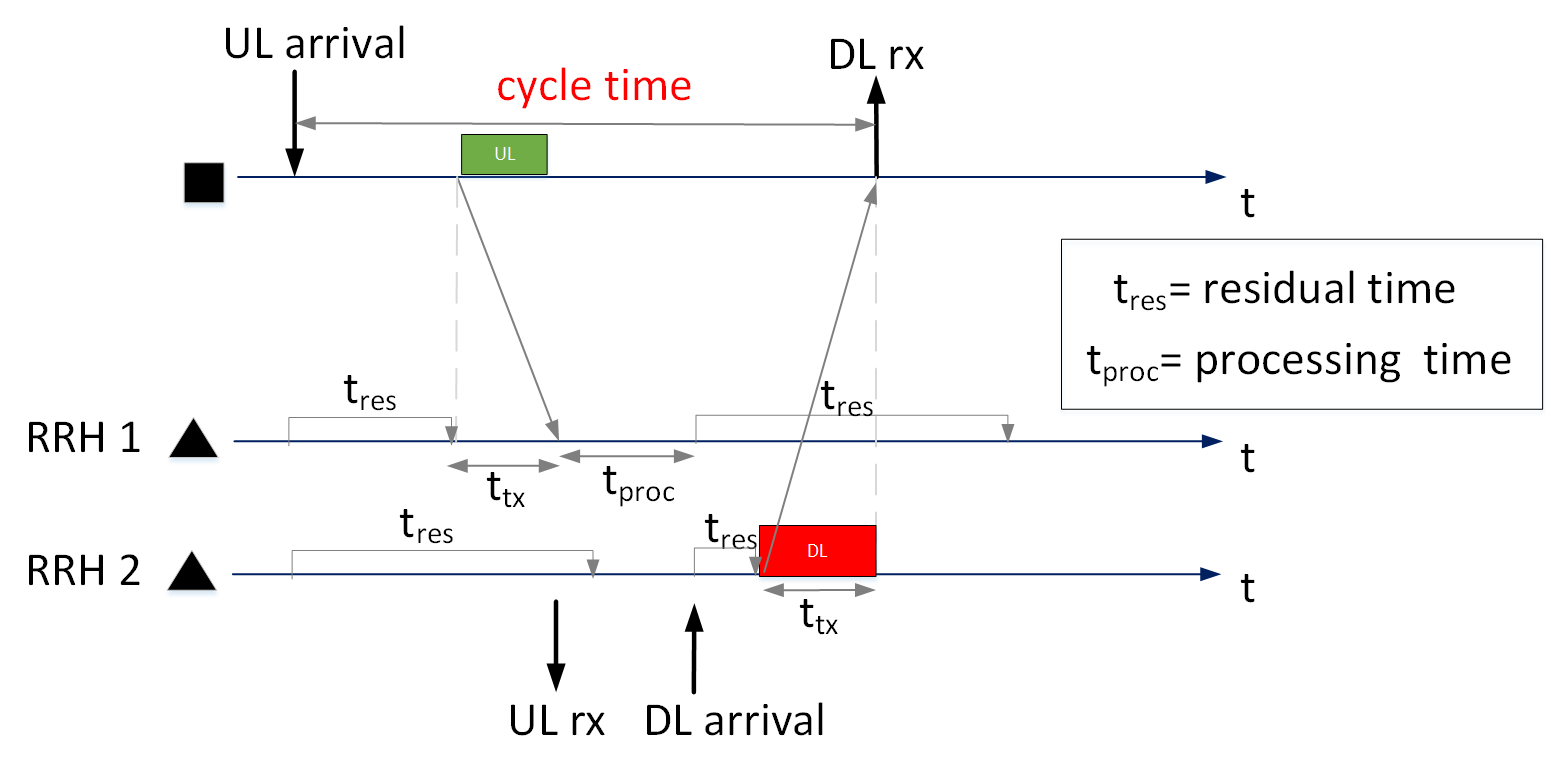}
	\caption{Sketch of the cycle time for a user with the highest priority using decoupled access. The RRH with the shortest residual time is selected for each transmission. }
	\label{fig:cycle_time}
\end{figure}

Figure \ref{fig:residual_time_exp} plots the CDF for the exemplary case of an exponential distribution, which corresponds to the residual time of an M/M/1 queue in the coupled access. In this case, the CDF yields
\begin{equation}
    F_Y(y) = 1-e^{2\lambda y} \;\;\; \forall 0\leq y \leq S_{L}
\end{equation} 

\section{Conclusions} \label{sec:conclusions}
We have investigated the latency gains of an interactive URLLC device when using flexible TDD and a decoupled UL/DL access. \colb{The critical URLLC traffic} is multiplexed with eMBB traffic, which usually requires much longer TTIs and adaptation to the instantaneous channel quality. The flexible TDD frame in 5G NR is the basis for the analysis. The problem is addressed from a queueing perspective, with the heterogeneous requirements and the Rayleigh channel variations captured in the model. The results show the latency improvements of the decoupled access, which are remarkable when the load increases. An upper bound of a priority user completes the analysis giving insight of the two-way round trip time. We have identified and quantified the potential of decoupling the two transmission directions, setting the basis for future work. Next steps include refining the model to include the impact of scheduling policies beyond FIFO. 

\begin{figure}[t]
	\centering
	\includegraphics[width=0.5\textwidth]{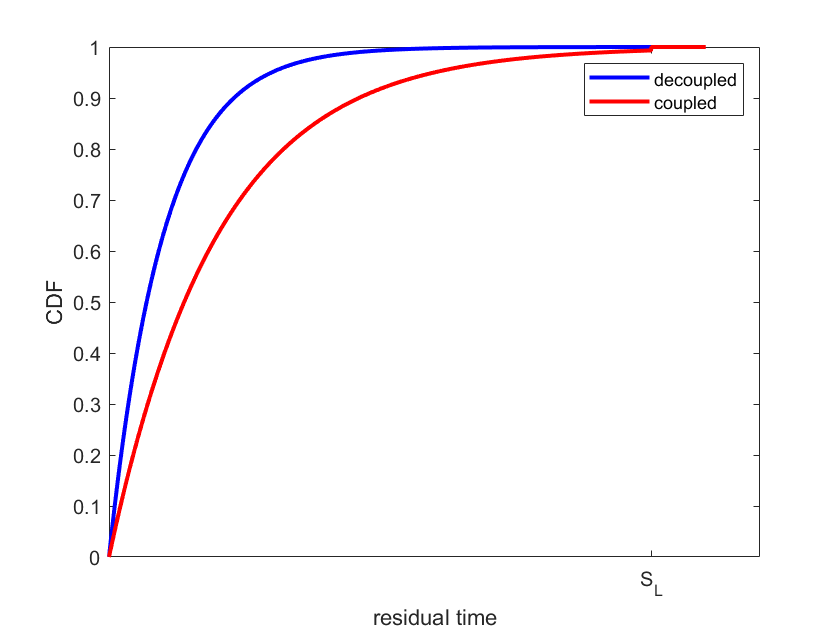}
	\caption{Decoupled latency gain for a critical URLLC user. The service time in each RRH is exponentially distributed and confined to $(0,S_L)$.}
	\label{fig:residual_time_exp}
\end{figure}

\section*{Acknowledgment}

This work has been in part supported by EU Horizon 2020 projects ONE5G (ICT-760809) and the European Research Council (ERC Consolidator Grant no.648382 WILLOW). The views expressed in this work are those of the authors and do not necessarily represent the ONE5G project view.

\bibliographystyle{IEEEtran}
\bibliography{QueueingAnalysisTDD}{}


\end{document}